\documentclass[runningheads,a4paper]{llncs}
\usepackage{amssymb}
\usepackage{graphicx}
\usepackage{color,xcolor}
\usepackage{amsmath,mathrsfs}
\usepackage{hyperref}

\usepackage{url}
\urldef{\mailsa}\path|{alfred.hofmann, ursula.barth, ingrid.haas, frank.holzwarth,|
\urldef{\mailsb}\path|anna.kramer, leonie.kunz, christine.reiss, nicole.sator,|
\urldef{\mailsc}\path|erika.siebert-cole, peter.strasser, lncs}@springer.com|

\makeatletter

\newcommand{\Rmnum}[1]{\expandafter\@slowromancap\romannumeral #1@}
\makeatother

\newcommand{\GF}[2][2]{\mathbb F_{{#1}^{#2}}}
\newcommand{\F}{\mathbb{F}}

\newcommand{\DDT}{\mathrm{DDT}}
\newcommand{\BCT}{\mathrm{BCT}}
\newcommand{\IM}{\mathrm{Im}}
\newcommand{\DU}{\Delta}
\newcommand{\BU}{\beta}
\newcommand{\Tr}{\mathrm{Tr}}
\newcommand{\U}{\mathcal{U}}

\begin{document}
\mainmatter
\title{On the boomerang uniformity of (quadratic) permutations over $\GF n$}
\titlerunning{On the boomerang uniformity of (quadratic) permutations over $\GF n$}

\author{Sihem Mesnager\inst{1}
\and Chunming Tang\inst{2} \and Maosheng Xiong\inst{3}
}

\authorrunning{S. Mesnager, C. Tang and M. Xiong}

\institute{LAGA, Department of Mathematics, Universities of Paris VIII and Paris XIII, CNRS, UMR 7539 and Telecom ParisTech, France\\
\email{smesnager@univ-paris8.fr}\\
\and
School of Mathematics  and Information,
China West Normal University, Nanchong 637002,  China, and  Department of Mathematics, The Hong Kong University
of Science and Technology, Clear Water Bay, Kowloon, Hong Kong\\
\email{tangchunmingmath@163.com}\\
\and
Department of Mathematics, The Hong Kong University of Science and Technology, Clear Water Bay, Kowloon, Hong Kong\\
\email{mamsxiong@ust.hk}}

\maketitle
\makebox[\linewidth]{\today}

\begin{abstract}

At Eurocrypt'18, Cid, Huang, Peyrin, Sasaki, and Song introduced a
  new tool called Boomerang Connectivity Table (BCT) for measuring the
  resistance of a block cipher against the boomerang attack (which is
  an important cryptanalysis technique introduced by Wagner in 1999
  against block ciphers). Next, Boura and Canteaut introduced an
  important parameter (related to the BCT) for cryptographic Sboxes
  called boomerang uniformity.  In this context, we present a brief
  state-of-the-art on the notion of boomerang uniformity of vectorial
  functions (or Sboxes) and provide new results. More specifically,
  we present a slightly different (and more convenient) formulation of the boomerang uniformity and show that the row sum and the column sum of the boomerang connectivity table can be expressed in terms of the zeros of the second-order derivative of the permutation or its inverse. Most importantly, we specialize our study of boomerang uniformity to quadratic permutations in even dimension and generalize the previous results  on quadratic permutation with optimal BCT (optimal means that the maximal value in the Boomerang Connectivity Table equals the lowest known differential uniformity). As a consequence of our general result, we prove that the boomerang uniformity  of the binomial differentially $4$-uniform permutations presented by Bracken, Tan, and Tan equals $4$. This result gives rise to a new family of optimal Sboxes.

  \end{abstract}


\section{Introduction}

Substitution boxes (Sboxes) are fundamental parts of block ciphers.
Being the only source of nonlinearity in these ciphers, they play a
central role in their robustness, by providing confusion.
Mathematically, Sboxes are vectorial (multi-output) Boolean functions,
that is, functions from the vector space $\mathbb{F}_{2}^n$ (of all
binary vectors of length $n$) to the finite field $\mathbb{F}_{2}^r$,
for given positive integers $n$ and $r$. These functions are called
$(n,r)$-functions and include the (single-output) Boolean functions
(which correspond to the case $r =1$). When they are used as S-boxes
in block ciphers, their number $r$ of output bits equals or
approximately equals the number $n$ of input bits. We shall identify
the vector space $\mathbb{F}_2^n$ to the finite field
$\mathbb{F}_{2^n}$ (of order $2^n$).
 A  nice survey on Boolean and vectorial Boolean functions for cryptography can be found in \cite{Cbook1} and \cite{Cbook}, respectively.

In 1999, Wagner \cite{Wagner1999} has introduced the {\em boomerang
  attack} which is an important cryptanalysis technique against block
ciphers involving Sboxes. These attacks can be seen as an extension of
classical differential attacks \cite{BS91}. In fact, they combines two
differentials for the upper part and the lower part of the cipher. The
dependency between these two differentials then highly affects the
complexity of the attack and all its variants (we refer for example to
\cite{BDK01,BDK02,BDD03,BK09,DKS10,KKS01,KHP+12} and the references
therein).

At Eurocrypt 2018, Cid, Huang, Peyrin, Sasaki, and Song
\cite{BOOMERANG-2018} introduced the concept of {\em Boomerang
  Connectivity Table} (BCT for short) of a permutation $F$ and we call
the maximum value in BCT of $F$ the boomerang uniformity of $F$. Such
a notion allows to simplify the complexity analysis, by storing and
unifying the different switching probabilities of the cipher's Sbox in
one table.  Very recently (2019),  Song, Qin and Hu \cite{SongQinHu} have revisited the BCT by proposing a generalized framework of BCT. They applied their new framework to two block ciphers SKINNY and AES which are two typical block ciphers using weak and strong round functions respectively.
In 2018, Boura and Canteaut
\cite{Boura-Canteaut-2018}, have introduced a parameter (related to
the BCT) for cryptographic Sboxes called {\em boomerang uniformity}
and provided a more in-depth analysis of boomerang connectivity tables
by studying more closely differentially 4-uniform Sboxes. They firstly
completely characterized the BCT of all differentially 4-uniform
permutations of 4 bits and then studied these objects for some
cryptographically relevant families of Sboxes, as the inverse function
and quadratic permutations. These two families provide the first
examples of differentially 4-uniform Sboxes optimal against boomerang
attacks for an even number of variables. Later, Li, Qu, Sun and Li
\cite{Li-etal2019} have essentially provided an equivalent definition
to compute the BCT and the boomerang uniformity, provided a
characterization of functions having a fixed boomerang uniformity by
means of the Walsh transform and finally exhibited a class of
differentially 4-uniform permutation for which the boomerang equals
$4$.  The aim of this manuscript is to increase our knowledge on boomerang uniformity of (quadratic) Sboxes in the line of the recent articles \cite{Boura-Canteaut-2018} and \cite{Li-etal2019}. \\

The paper is structured as follows. In Section \ref{Preliminaries}, we introduce the needed definitions related to the  differential uniformity and boomerang uniformity of vectorial functions and briefly discuss these notions.
In Section \ref{boomerang}, we present shortly the state-of-the-art on  boomerang uniformity of Sboxes, we present a slightly different and  more convenient formulation of the boomerang uniformity and show that the row sum and the column sum of the boomerang connectivity table can be expressed in terms of the zeros of the second-order derivative of the permutation or its inverse.
Next, in Section \ref{boomerang_uniformity_quadratic_permutations}, we specialize our study of boomerang uniformity to quadratic permutations in
even dimension and generalize the results presented in \cite{Boura-Canteaut-2018} and \cite{Li-etal2019} on quadratic permutation with optimal BCT (optimal means that the maximal value in the Boomerang Connectivity Table equals the lowest known differential uniformity). Consequently, we recover all the known results and prove that the boomerang uniformity  of the binomial differentially $4$-uniform permutations presented by Bracken, Tan, and Tan equals $4$. This result gives rise to a new family of optimal Sboxes.

\section{Preliminaries}\label{Preliminaries}

In this section, we recall some notations, definitions and results related to the differential properties and boomerang uniformity of functions.

Throughout this article, $\# E$ denotes the cardinality of a finite set $E$. The binary field is denoted by $\mathbb{F}_2$ and the finite field of order $2^n$ (resp. q) is denoted by $\mathbb{F}_{2^n}$ (resp. $\mathbb{F}_{q}$). The multiplicative group ${\Bbb F}^*_{2^n}$ is a cyclic group consisting
of $2^n-1$ elements. The terminology Sbox refer to an $(n,n)$-vectorial function, that is, a function from $\GF{n}$ to itself.

Any function $F$ from $\GF{n}$ to itself admits a (unique) representation as a polynomial over~$\Bbb{F}_{2^n}$
in one variable and of (univariate) degree at most $2^n-1$:\begin{equation}\label{ur}F(x)=\sum_{i=0}^{2^n-1}\delta _i x^i; \quad \delta_i\in {\Bbb F}_{2^n}.\end{equation}

For any $k$, $0\leq k \leq 2^n-1$, the number $w_2(k)$ of the nonzero coefficients $k_s\in\{0,1\}$ in  the binary expansion of $k$ is called the $2$-weight  of $k$.  The algebraic degree of  $F$ is equal to the maximum $2$-weight of the exponents $i$ of the polynomial $F(x)$ such that $\delta _i \not=0$, that  is,  $deg (F)= max_{0\leq i\leq n-1, \delta _i \not=0}w_2(i)$. A function $F$ from $\GF{n}$ to itself is said to be quadratic if $deg(F)=2$.

We use the following terminology: a permutation polynomial over $\mathbb{F}_q$ is a polynomial $F(x)\in \mathbb{F}_{q}[x]$ for which the function $a\mapsto F(a)$ defines a permutation of $\mathbb{F}_q$.

Recall that for any positive integers $k$ and $r$ such that $r|k$, the trace function from $\mathbb{F}_{2^k}$ to $\mathbb{F}_{2^r}$, denoted by $\Tr_{r}^{k}$, is the mapping defined as
\begin{displaymath}
  \Tr_{r}^{k}(x):=\sum_{i=0}^{\frac kr-1}
  x^{2^{ir}}=x+x^{2^r}+x^{2^{2r}}+\cdots+x^{2^{k-r}}.
\end{displaymath}
In particular, the {\em absolute trace} over $\mathbb{F}_2$ of an element $x \in
\mathbb{F}_{2^n}$ equals $\Tr_1^{n}(x)=\sum_{i=0}^{n-1} x^{2^i}$.

\begin{definition}
Given an $(n,n)$-function $F$, the derivative of $F$ with respect to $a\in\GF n$ is the function $D_aF: x \mapsto F(x+a)+F(x)$. For $(a,b)\in\left(\GF n\right)^2$,  the second order derivative of $F$ with respect to $a\in\GF n$ and $b\in\GF n$  is the function $D_a D_bF: x \mapsto F(x+a)+F(x+b)+F(x+a+b)+F(x)$. For $(a,b)\in\left(\GF n\right)^2$,
 the entries of the difference distribution table (DDT) are given by
\begin{displaymath}
  \DDT_F(a,b) = \#\{x\in\GF n\mid D_aF(x) = b\}.
\end{displaymath}
The differential uniformity of $F$ is defined as
\begin{displaymath}
  \DU(F) = \max_{a,b\in\GF n^\star} \DDT_F(a,b).
\end{displaymath}
\end{definition}

Differential uniformity is an important concept in cryptography as it quantifies the degree of security of
a Substitution box used in the cipher with respect to differential attacks. APN (Almost Perfect Nonlinear) functions $F$ are those such that their differential
uniformity equals $2$ (i.e. $\DU(F)=2$).

\begin{definition}
Let $F$ be a permutation of $\GF n$. For
$(a,b)\in\left(\GF n\right)^2$,  we define the entries of the {\em boomerang connectivity table }(BCT) as
\begin{displaymath}
  \BCT_F(a,b) = \#\{x\in\GF n\mid F^{-1}(F(x)+b) + F^{-1}(F(x+a)+b) = a\},
\end{displaymath}
where $F^{-1}$ denotes the compositional inverse of $F$.
The boomerang uniformity of $F$ is defined as
\begin{displaymath}
  \BU(F) = \max_{a,b\in\GF n^\star} \BCT_F(a,b).
\end{displaymath}
\end{definition}


\noindent Observe that $\DDT_F(a,b)$ is equal to $0$ or $2^n$ when $ab=0$. Likewise,
$\BCT_F(a,b)=2^n$ when $ab=0$.

It is well-known that $F$ and $F^{-1}$ have the same differential
uniformity since
\begin{eqnarray*}
  F^{-1}(x) + F^{-1}(x+a) = b
  &\iff& x + a = F(F^{-1}(x) + b) \\
  &\iff& F(F^{-1}(x)) + F(F^{-1}(x) + b) = a,
\end{eqnarray*}
yielding that $\DDT_{F^{-1}}(a,b) = \DDT_F(b,a)$ for any
$(a,b)\in(\GF n^\star)^2$. Not surprising, as shown in \cite[Proposition
2]{Boura-Canteaut-2018}, it is also true for the boomerang uniformity
since
\begin{eqnarray*}
  & & F^{-1}(F(x)+b) + F^{-1}(F(x+a)+b) = a \\
  && \iff F(x) + F(F^{-1}(F(x+a)+b) + a) = b\\
  && \iff F(F^{-1}(F(x+ a)) + a) + F(F^{-1}(F(x+a)+b) + a) = b,
\end{eqnarray*}
that is,
\begin{eqnarray} \label{2:boom} \BCT_F(a,b) = \BCT_{F^{-1}}(b,a), \quad \forall (a, b) \in (\GF n^\star)^2. \end{eqnarray}

\section{On boomerang uniformity of Sboxes}\label{boomerang}

For vectorial Boolean functions, the most useful concepts of
equivalence are the extended affine EA-equivalence and the
CCZ-equivalence. Two $(n,r)$-functions $F$ and $F^{\prime}$ are called
EA-equivalent if there exist affine permutations $L$ from $\GF n$ to
$\GF n$ and $L^{\prime}$ from $\GF r$ to $\GF r$ and an affine
function $L''$ from $\GF n$ to $\GF r$ such that
$F'=L^{\prime} \circ F \circ L + L''$. EA-equivalence is a particular
case of CCZ-equivalence \cite{CCZ98}. Two $(n,r)$-functions $F$ and
$F^{\prime}$ are called CCZ-equivalent if their graphs
$G_F:=\{(x,F(x)),~ x\in \GF n$\} and
$G_F^{\prime}:=\{(x,F^{\prime}(x)),~ x\in \GF n$\} are affine
equivalent, that is, if there exists an affine permutation
$\mathcal{L}$ of $\GF n \times \GF r$ such that
$\mathcal{L}(G_F)=G_F^{\prime}$.

As explained in \cite{Boura-Canteaut-2018}, the multi-set formed by
all values in the BCT is invariant under affine equivalence and
inversion. In other words, the behaviour of the BCT with respect to
these two classes of transformations is exactly the same as the
behaviour of the DDT.  However, while the differential spectrum of a
function is also preserved by the extended affine (EA) equivalence,
this is not the case for the BCT. As EA-equivalence is a special case
of CCZ equivalence, the boomerang uniformity is also not always
preserved under CCZ-equivalence.

In \cite{BOOMERANG-2018}, it was indicated that $\BCT(a,b)$ is greater than or
equal to $\DDT_F(a,b)$ yielding
\begin{theorem}(\cite{BOOMERANG-2018})\label{lower_bound}
  Let $F$ be a permutation of $\GF n$. Then, $\BU(F)\geq\DU(F)$.
\end{theorem}
It was also proved that
\begin{theorem}(\cite{BOOMERANG-2018})\label{APN}
  Let $F$ be a permutation of $\GF n$. Then, $\DU(F)=2$ if and only if
  $\BU(F)=2$.
\end{theorem}

In \cite{Boura-Canteaut-2018}, Boura and Canteaut established an alternative formulation of the boomerang uniformity as follows:

\begin{theorem}(\cite{Boura-Canteaut-2018})\label{boura-canteaut}
  Let $F$ be a permutation of $\GF n$. Then, for any $a$ and $b$ in $\GF n^\star$
  \begin{displaymath}
    \BCT_F(a,b) = \DDT_F(a,b) + \sum_{\gamma\in\GF n^\star,\gamma\not=b}\#\left(\mathcal U_{\gamma,a}^{F^{-1}}\cap\left(b + U_{\gamma,a}^{F^{-1}}\right)\right)
  \end{displaymath}
  where
  \begin{displaymath}
    \mathcal U^{F^{-1}}_{\gamma,a} = \{x\in\GF n\mid D_{\gamma}F^{-1}(x) = a\}.
  \end{displaymath}
\end{theorem}

In this paper, we shall use a slightly different formulation of Theorem~\ref{boura-canteaut}. 

\begin{theorem} \label{boom-new-form}
  Let $F$ be a permutation of $\GF{n}$. Then, for any $a$ and $b$ in $\GF n^\star$
  \begin{eqnarray}\label{boom-eq1}
    \BCT_F(a,b) = \sum_{\gamma\in\GF n^\star}\#\left(\mathcal U_{\gamma,b}^{F}\cap\left(a + \U_{\gamma,b}^{F}\right)\right),
  \end{eqnarray}
  where
  \begin{displaymath}
    \mathcal U^{F}_{\gamma,b} = \{x\in\GF n\mid D_{\gamma}F(x) = b\}.
  \end{displaymath}
\end{theorem}
\begin{proof} Observing that $\BCT_F(a,b)=\BCT_{F^{-1}}(b,a)$, we can prove Equation (\ref{boom-eq1}) by directly applying Theorem \ref{boura-canteaut}. Here we mention another method. It was proved in \cite{Li-etal2019} that the quantity $\BCT_F(a,b)$ equals the number of solutions $(x,y) \in (\GF{n})^2$ satisfying the equations $F(x)+F(y)=b$ and $F(x+a)+F(x+a)=b$ simultaneously. Letting $x+y=\gamma$, then $y=x+\gamma$, so we have
\begin{eqnarray*}
\BCT_F(a,b) &=& \#\left\{ (x,\gamma):
\begin{array}{c}
F(x)+F(x+\gamma)=b \\
F(x+a)+F(x+\gamma+a)=b \end{array} \right\}\\
&=& \sum_{\gamma} \#\left\{
x: \begin{array}{c}
F(x)+F(x+\gamma)=b \\
F(x+a)+F(x+\gamma+a)=b
\end{array}\right\}.
\end{eqnarray*}
It is easy to see that the set in the inner sum for each subscript $\gamma$ is $\U_{\gamma,b}^{F}\cap\left(a + \U_{\gamma,b}^{F}\right)$. Moreover, $\U_{0,b}^{F} = \emptyset$ since $F$ is a permutation. This completes the proof of Theorem \ref{boom-new-form}.
\qed
\end{proof}

The new formulation (\ref{boom-eq1}) seems more convenient to use than Theorem \ref{boura-canteaut} which involves $F^{-1}$. Moreover, in (\ref{boom-eq1}), the condition that $F$ is a permutation is not required, that is, we may define the boomerang uniformity for any $(n,n)$ function $F$, even though it may not be a permutation. This is similar to the concept of differential uniformity, which may be of future interest. Finally, in (\ref{boom-eq1}), since $\U_{a,b}^{F} =a + \U_{a,b}^{F}$, we have
\[\BCT_F(a,b) = \DDT_F(a,b)+ \sum_{\gamma \ne a,0}\#\left(\mathcal U_{\gamma,b}^{F}\cap\left(a + \U_{\gamma,b}^{F}\right)\right), \]
from which we can immediately derive $\BCT_F(a,b) \ge \DDT_F(a,b)$.

In the following, we  show that the row sum and the column sum of the boomerang connectivity table can be expressed in terms of the zeros of the second-order derivative of the permutation or its inverse.
\begin{proposition}\label{row-and-column-sum}
For any $a$ and $b$ in $\GF n^\star$, we have
\begin{eqnarray} \label{3:bct1}
\sum_{c\in\GF n^\star}\BCT_F(a,c) &=& \sum_{c\in\GF n^\star}\#\{x\in\GF n\mid D_aD_cF(x) = 0\} \\
\label{3:bct2} &=& \sum_{c\in\GF n^\star} \DDT_F(a,c)^2-2^n,
  \end{eqnarray}
  and
  \begin{eqnarray*}
    \sum_{c\in\GF n^\star}\BCT_F(c,b) &=& \sum_{c\in\GF n^\star}\#\{x\in\GF n\mid D_bD_cF^{-1}(x) = 0\}\\
    &=& \sum_{c\in\GF n^\star} \DDT_F(c,b)^2-2^n.
  \end{eqnarray*}
\end{proposition}
\begin{proof}
Let $a$ and $b$ be in $\GF n^\star$.  According to (\ref{boom-eq1})
  \begin{eqnarray*}
    \sum_{c\in\GF n^\star}\BCT_F(a,c)
    &=& \sum_{c\in\GF n^\star}\sum_{\gamma\in\GF n^\star} \#\{x\in\GF n\mid D_\gamma F(x) = D_\gamma F(x+a) = c\}\\
    &=& \sum_{\gamma\in\GF n^\star}  \#\{x\in\GF n\mid D_\gamma F(x) = D_\gamma F(x+a)\}.
  \end{eqnarray*}
Here we have used the fact that $D_{\gamma}F(x) \ne 0$ for any $x \in \GF n$ and $\gamma \in \GF n^\star$ since $F$ is a permutation. Now,
  the identity (\ref{3:bct1}) follows immediately from the fact that
  $D_\gamma F(x) = D_\gamma F(x+a)$ if and only if $D_aD_\gamma F(x) = 0$.

As for (\ref{3:bct2}), first we observe
\begin{eqnarray*}
&&\sum_{c\in\GF n}\BCT_F(a,c)\\
&=& \sum_{c \in \GF n} \#\left\{(x,y) \in (\GF n)^2: \begin{array}{l}
F(x)+F(y)=c \\
F(x+a)+F(y+a)=c\end{array} \right\} \\
&=&
\#\left\{(x,y) \in (\GF n)^2: \begin{array}{l}
F(x)+F(y)=F(x+a)+F(y+a)\end{array} \right\}\\
&=& \sum_{c \in \GF n} \#\left\{(x,y) \in (\GF n)^2: \begin{array}{l}
F(x)+F(x+a)=c \\
F(y)+F(y+a)=c\end{array} \right\}=\sum_{c\in\GF n} \DDT_F(a,c)^2.
\end{eqnarray*}
Noting that $\BCT_F(a,0)=2^n$ and $\DDT_F(a,0)=0$, we obtain the desired identity. The proof of the second assertion is similar, by using $\BCT_F(c,b) = \BCT_{F^{-1}}(b,c)$. This completes the proof of Theorem \ref{row-and-column-sum}.
\qed
\end{proof}

\section{On the boomerang uniformity of quadratic permutations}\label{boomerang_uniformity_quadratic_permutations}

In this section, we specialize our study of boomerang uniformity to quadratic permutations in
even dimension. In even dimension, it is known that the best differential uniformity of a quadratic permutation is $4$. Note that in even dimension, APN quadratic permutations $F$ do not exist \cite{Nyb95}. Therefore, $\Delta(F)=4$ is the lowest differential uniformity that a quadratic permutation can achieve in this case.

In \cite{Boura-Canteaut-2018}, Boura and Canteaut exhibited a family of quadratic
permutations with optimal BCT, i.e. permutations which have differential uniformity
and boomerang uniformity both equal to $4$. In \cite{Li-etal2019} Li et al. provided another family of quadratic permutations with optimal BCT. These two families of permutations are related to the so-called Gold power permutations \cite{Gol68}. In this section we obtain a vast generalization of these results.

\begin{theorem} \label{3:thm1}
Let $q$ be a power of $2$ and $m$ a positive integer. Let $F: \F_{q^m} \to \F_{q^m}$ be a quadratic function of the form
\begin{eqnarray} \label{4:qform}
F(x)=\sum_{0\leq i\leq j\leq m-1}c_{ij}x^{q^i+q^j}, \quad  \forall c_{ij} \in \mathbb{F}_{q^m}.
\end{eqnarray}
Then $\DU(F) \ge q$. Moreover, if $F$ is a permutation on $\F_{q^m}$ and $\DU(F)=q$, then $\beta(F)=q$.
\end{theorem}

\begin{proof}

For any $\gamma \in \F_{q^m}^\star$, let
\[H_{\gamma}(x)=F(x+\gamma)+F(x)+F(\gamma). \]
We have
  \begin{align*}
    H_{\gamma}(x)=& \sum_{0\leq i\leq j \leq m-1 }c_{ij}\left ( (x+\gamma)^{q^i+q^j}+x^{q^i+q^j}+ {\gamma}^{q^i+q^j} \right ) \\
    =& \sum_{0\leq i < j\leq m-1 }
    c_{ij}(x^{q^i} {\gamma}^{q^j}+x^{q^j}{\gamma}^{q^i}).
  \end{align*}
We apply Theorem \ref{boom-new-form}. The set $\U_{\gamma,b}^F$ is given by
\begin{eqnarray*} \U_{\gamma,b}^F = \left\{x \in \F_{q^m}: H_{\gamma}(x)=b+F(\gamma)\right\}. \end{eqnarray*}
Since $F$ is a quadratic function, if $\U_{\gamma,b}^F \ne \emptyset$, that is, $\DDT_F(\gamma,b)>0$, then $\U_{\gamma,b}^F$ is an affine subspace of $\F_{q^m}$ obtained by a translation of the vector space $K_F(\gamma) $ where
\[K_F(\gamma)=\left\{x \in \F_{q^m}: H_{\gamma}(x)=0 \right\}. \]
It is easy to see that $K_F(\gamma)$ is a vector space over $\F_q$ and $\gamma \cdot \F_q \subset K_F(\gamma)$. Thus we have $\DU(F) \ge \# K_F(\gamma) \ge q$.

Now suppose that $F$ is a permutation and $\DU(F)=q$. Then we have
\[K_F(\gamma)= \gamma \cdot \F_q. \]
Thus
\[ \U_{\gamma,b}^F \cap (a+\U_{\gamma,b}^F) = \left\{\begin{array}{cll}
\emptyset &:& \mbox{ if } \DDT_F(\gamma,b)=0 \mbox{ or } a \notin \gamma \cdot \F_q,\\
\U_{\gamma,b}^F &:& \mbox{ if } \DDT_F(\gamma,b)>0 \mbox{ and } a \in \gamma \cdot \F_q,\end{array}\right.\]
and from Theorem \ref{boom-new-form} we have
\begin{eqnarray} \label{4:gboom} \BCT_F(a,b)=\sum_{\gamma \in \F_q^\star} \DDT_F(a \gamma,b). \end{eqnarray}

Now suppose there exist $\gamma_1 \ne \gamma_2 \in \F_q^\star$ such that
\[\DDT_F(a \gamma_i,b) \ne 0, \quad \forall i=1,2, \]
then the equations
\[H_{a \gamma_i}(x)=b+F(a \gamma_i), \quad x \in \F_{q^m}\]
are solvable for both $i=1,2$. Noting that for $i=1,2$, since $\gamma_i \in \F_q$,
\[H_{a \gamma_i}(x)= \gamma_i H_a(x), \quad F(a \gamma_i)=\gamma_i^2F(a), \]
we have
\begin{eqnarray*} \frac{b}{\gamma_i}+\gamma_i F(a) \in \IM H_a, \quad i=1,2,  \end{eqnarray*}
where $\IM H_a$ is the image of the function $H_a(x)$ on $\F_{q^m}$. It is easy to see that the set $\IM H_a$ is a vector space over $\F_q$. Thus we have
\[b+\gamma_i^2 F(a) \in \IM H_a, \quad i=1,2, \]
and
\[(b+\gamma_1^2F(a))-(b+\gamma_2^2F(a))= (\gamma_1-\gamma_2)^2 F(a) \in \IM H_a,\]
and hence $F(a) \in \IM H_a$. However, this means that the equation
\[F(x+a)+F(x)=0\]
is solvable for $x \in \F_{q^m}$, which is impossible because $F$ is a permutation and $a \ne 0$.

So we have proved that in (\ref{4:gboom}) there is at most one $\gamma \in \F_q^\star$ such that $\DDT_F(a \gamma,b) \ne 0$. Noting that actually
\[\DDT_F(a,b) \in \left\{0,q\right\} \quad \forall (a,b) \in (\F_{q^m}^\star)^2, \]
we obtain $\BCT_F(a,b) \le q$ for any $(a,b) \in (\F_{q^m}^\star)^2$. So we conclude that $\beta(F)=q$. This completes the proof of Theorem \ref{3:thm1}.
\qed
\end{proof}

We remark that a general quadratic function $F(x)$ of the form (\ref{4:qform}) may be written as
\[F(x)=f(x)+\phi(x),\]
where
\[f(x)=\sum_{0\leq i < j\leq m-1}c_{ij}x^{q^i+q^j}, \quad \phi(x)=\sum_{0\leq i\leq m-1}c_{i}x^{2q^i}. \]
Noting that $\phi$ is linear, so $\DU(F)=\DU(f)$.

As applications of Theorem \ref{3:thm1}, we show here how two previous results of quadratic permutations with optimal BCT from \cite{Boura-Canteaut-2018} and \cite{Li-etal2019} can be obtained.
\begin{enumerate}
\item[(1).] Let $n \equiv 2 \pmod{4}$ and let $t$ be an even integer such that $\gcd(t,n)=2$. It was known that  the function $F(x)=x^{2^t+1}$ is a permutation on $\GF n$ and $\DU(F)=4$ (see \cite{BCC-2010}). Theorem  \ref{3:thm1} implies that $\beta(F)=4$. This is \cite[Proposition 8]{Boura-Canteaut-2018}.

\item[(2).] Let $n=2m$ where $m$ is an odd integer. Let $\gamma \in \GF n^{\star}$ be an element such that the order of $\lambda^{2^m-1}$ is 3. Define $F(x)=x^{2^m+2}+\lambda x$. It was known that $F(x)$ is a permutation on $\GF n$ with $\DU(F)=4$ \cite{Ziv}. We may write $F(x)=f(x)^2$ where $f(x)=x^{2^{m-1}+1}+\lambda^{2^{-1}} x^{2^{n-2}+2^{n-2}}$, hence Theorem \ref{3:thm1} implies that $\beta(F)=4$. This is \cite[Theorem 5.3]{Li-etal2019}.
\end{enumerate}

Next we exhibit another family of quadratic permutations with optimal BCT.
To this end, we mention that it has been found in
\cite{BRACKEN2012537} a highly nonlinear $4$-differential uniform
permutation of $\GF n$ with $n=3k$, $k \equiv 2 \pmod{4}$, $3 \nmid k$:
\begin{equation}
  \label{eq:Bracken-Tan-Tan}
  F(x) = \beta x^{2^s+1} + \beta^{2^{k}} x^{2^{-k}+2^{k+s}}
\end{equation}
where $\gcd(n,s)=2$, $3\vert k+s$ and $\beta$ is a primitive element
of $\GF n$. It was proved that $\DU(F)=4$. So from Theorem \ref{3:thm1} we obtain

\begin{corollary}
 Let $F$ be  an $(n,n)$-function (permutation) defined by (\ref{eq:Bracken-Tan-Tan}). Then, $\BU(F)=4$.
\end{corollary}

Finally, we show by numerical computation that many new quadratic permutations with optimal BCT can be found by Theorem \ref{3:thm1}. Here we focus only on the quadratic function $F(x)$ of the form
\begin{eqnarray} \label{4:form} F(x)=x^{2^{s+1}+2}+Ax+Bx^4+C x^{16}, \quad A,B,C \in \GF n,\end{eqnarray}
where $n \equiv 2 \pmod{4}$ and $\gcd(n,s)=2$. It was known that $\DU(F)=4$. Noting that
$F(x)=f(x)^2$ where
\[f(x)=x^{2^{s}+1}+A^{2^{-1}}x^{2^{n-1}}+B^{2^{-1}}x^2+C^{2^{-1}} x^{8}, \quad A,B,C \in \GF n, \]
by Theorem \ref{3:thm1}, if $F$ is a permutation, then $\beta(F)=4$. For simplicity we only consider the case $n=6$ and search via {\bf Magma} triples $(A,B,C) \in (\GF 6)^3$ such that $F$ is a permutation, the number of which is given in the tables below. We also indicate the number of such $F$'s implied by \cite[Proposition 8]{Boura-Canteaut-2018} and \cite[Theorem 5.3]{Li-etal2019}. This really shows that there is a abundance of quadratic permutations with optimal BCT.

\begin{table}[ht]
\begin{center}
\caption{The number of quadratic permutations $F(x)$ of the form (\ref{4:form}) on $\GF 6$ with optimal BCT}\label{table}
\begin{tabular}{|c|c|c|c|} \hline
 & Theorem \ref{3:thm1} (this paper) & \cite[Proposition 8]{Boura-Canteaut-2018} &  \cite[Theorem 5.3]{Li-etal2019}  \\ \hline
$s=2$ &  960  & 1 & 15 \\ \hline
\end{tabular}
\end{center}
\end{table}


{\bf Acknowledgement.} The authors thank Cunsheng Ding, Nian Li and Haode Yan for their interesting discussions.

\end{document}